\documentclass[submission,copyright,creativecommons]{eptcs}
\providecommand{\event}{WPTE18 Postproceedings} % Name of the event you are submitting to
\usepackage{breakurl}             % Not needed if you use pdflatex only.
\usepackage{amsmath}
\usepackage{amsthm}
\usepackage{amssymb}
\usepackage{esvect}
\usepackage{paralist}
\usepackage[all]{xy}
\usepackage{tikz}
\usetikzlibrary{positioning,arrows}
\usepackage{pgfplots}
\usepackage{longtable}
\title{Optimizing Space of Parallel Processes}
\author{Manfred Schmidt-Schau{\ss}\thanks{supported by the Deutsche Forschungsgemeinschaft (DFG) under grant SCHM 986/11-1.}
\institute{Goethe-University\\Frankfurt am Main}
\email{schauss@ki.cs.uni-frankfurt.de}
\and
 Nils Dallmeyer\footnotemark[1]
\institute{Goethe-University\\Frankfurt am Main}
\email{dallmeyer@ki.cs.uni-frankfurt.de} }
\def\titlerunning{Optimizing Space of Parallel Processes}
\def\authorrunning{M. Schmidt-Schau{\ss} \& N. Dallmeyer}
\usepackage{latexsym}
\usepackage{xcolor}
\usepackage{amsmath}
\usepackage{paralist}
\usepackage{esvect}
\usepackage{verbatim}
\usepackage{extarrows}
\usepackage{amsfonts,amssymb,mathtools}
\usepackage[all]{xy}
\usepackage{multirow}
\newtheorem{theorem}{Theorem}[section]

\newtheorem{proposition}[theorem]{Proposition}
\newtheorem{lemma}[theorem]{Lemma}
\newtheorem{corollary}[theorem]{Corollary}

\newtheorem{example}[theorem]{Example}
\newtheorem{definition}[theorem]{Definition}

\newtheorem{algorithm}[theorem]{Algorithm}

\newcommand{\PT}{\mathit{PT}}
\newcommand{\PAR}{{\ensuremath{\hspace{0.1mm}{\!\scalebox{2}[1]{\tt |}\!}\hspace{0.1mm}}}}
\newcommand{\MVAR}[2]{{#1}\,\textbf{\textsf{m}}\,{#2}}
\newcommand{\EMPTYMVAR}[1]{\MVAR{{#1}}{-}}
\newcommand{\NEW}{\nu}
\newcommand{\THREAD}[2]{{#1}\,{\Leftarrow}\,{#2}}
\newcommand{\tletr}{\text{\normalfont\ttfamily letrec}}
\newcommand{\tof}{\text{{\normalfont\ttfamily of}}}
\newcommand{\tTrue}{\text{{\normalfont\ttfamily True}}}
\newcommand{\tin}{\text{{\normalfont\ttfamily in}}}
\newcommand{\tseq}{\text{{\normalfont\ttfamily seq}}}
\newcommand{\tcase}{\text{{\normalfont\ttfamily case}}}
\newcommand{\talt}{{\tt alt}}
\newcommand{\casepf}{\,\texttt{->}\,}
\newcommand{\treturn}{\text{\normalfont\ttfamily return}}
\newcommand{\tfuture}{\text{{\normalfont\ttfamily future}}}
\newcommand{\MTHREAD}[2]{{#1} \xLongleftarrow{\!\!\text{\normalfont\sffamily main}\!\!}{#2}}
\newcommand{\maycon}{{\downarrow}}
\newcommand{\maydiv}{{\uparrow}}
\newcommand{\CHF}{\mathit{CHF}}

\newcommand{\CHFSTGC}{\ensuremath{\mathit{CHF}^*\mathit{GC}}}
\newcommand{\size}{\mathtt{size}}
\newcommand{\spmin}{\mathit{spmin}}
\newcommand{\itMax}{\mathit{Max}}
\newcommand{\spm}{\mathit{sps}}
\newcommand{\append}{{\texttt{++}}}
\newcommand{\cons}{{\texttt{:}}}
\newcommand{\SPOPTn}{{\sc SpOptN}}
\newcommand{\tdo}{{\tt do}}
\newcommand{\spet}{interleaving}
\newcommand{\spets}{interleavings}

\begin{document}

\maketitle
\begin{abstract}
This paper is a contribution to exploring and analyzing space-improvements in concurrent programming languages,
in particular in the functional process-calculus CHF.
Space-improvements are defined as a generalization of the corresponding notion in deterministic pure functional languages. 
The main part of the paper is the $O(n\cdot\log n)$ algorithm \SPOPTn{} for offline space optimization of several 
parallel independent processes.
Applications of this algorithm are: (i) affirmation of space improving transformations for particular classes of program transformations;
(ii) support of an interpreter-based method for refuting space-improvements; and 
(iii) as a stand-alone offline-optimizer for space (or similar resources) of parallel processes. 
 
\end{abstract}

\vspace*{0.2cm}
\hphantom{a} \hspace*{2.5cm} concurrency

\section{Introduction} 

The main motivation for investigating the common space consumption of parallel processes is our investigation into
 space optimizations and space improvements in concurrent languages. 
   A special but important subcase are parallel processes (threads) which are independent or have only rare interactions by a controllable form of synchronization.
  An algorithm to  compute a space-minimal execution sequence of a set of given parallel 
   and independent processes would be a first step in space optimizations and of great help for the analysis of space-improvements and -optimizations of programs. 
 
 The space consumption  of threads that are evaluated in parallel is as follows.
  We assume that there is a common memory, 
 where the state of every process is stored. In addition we assume that the storage occupation of processes is independent of each other.
 The model for processes is rather abstract insofar as it only models the thread-local space as a sequence of numbers.  
 Note that even in the case of only two independent threads the naive computation of the minimally necessary (thread-local) space to run the two threads 
 leads to an exponential number of different schedules, which cannot be checked by a brute force search.  
 As we will demonstrate in this paper, a deeper analysis shows that for independent processes (without communication, with the exception of the start and end), 
 this minimum can be computed with an offline-algorithm 
 in  time $O((N+n)\log N)$ where $N$ is the number of processes and $n$ the size of the input (Theorem \ref{theorem:n-processes-polynomial}).
 The prerequisite for the algorithm is that the complete space trace  of every single sub-process is already given, insofar the optimization can be classified as offline. 
 Our abstract model can be applied if all processes have a common start and end time.  
 
 This simplicity of our model invites applications of the space-optimization algorithm also for
 \begin{itemize}
   \item industrial processes (jobs) where the number of machines can be optimized since it is similar to required space
      (resource-restricted scheduling). It can  be used in problem settings similar to  job-shop-scheduling problems \cite{garey-jsp:76}, where the number of machines has 
      to be minimized and 
 where the time is not relevant (see e.g. \cite{garey-johnson-two-mach:77}). 
   \item  (independent) concurrent threads, independent  of a programming language.
 \end{itemize}
 
 Our  model is also extended to synchronization constraints in the form of a Boolean combination of  conditions on simultaneous  and/or relative time points of  two threads.
 The results for the space optimization for synchronization-free processes can be transferred to processes with synchronizations and permits
 polynomial algorithms for a fixed number of synchronization constructs 
(see Theorem \ref{theorem:upper-complexity-with-restrictions})
and therefore allows further analyses of space in more concrete scenarios.
In general, i.e. for arbitrary Boolean constraints, finding the minimum is NP-complete (Theorem \ref{thm:NP-hard}).

 The concrete  programming language model that we investigate is the functional process calculus CHF,  a variant of Concurrent Haskell, 
 which  permits pure and declarative functional modelling in 
combination with sequential 
(monadic) execution of processes with 
synchronization and which employs lazy evaluation  \cite{haskell-web:2016,peyton-gordon-finne:96,sabel-schauss-ppdp11:2011}.
  Related work on space improvements in  deterministic  call-by-need functional languages is \cite{gustavsson-sands:99,gustavsson-sands:01,schmidt-schauss-dallmeyer-wpte:17}.

An application of  results and algorithms for the space-minimization task in special cases is on the one hand to identify program transformation as space improvements (in CHF) and 
on the other hand to accelerate an automated search for potential counterexamples to conjectures of 
 space-improvements.
 Space optimization of parallel processes can sometimes be also applied to  CHF-programs. For example for processes that are deterministically parallel, 
 i.e. there is no sharing between processes, no free variables and the computation terminates.
In these special cases the notion of space improvement is the same as space optimization.

The {\bf structure of the paper} is first to informally explain the functional process calculus $\CHFSTGC$ and a definition of a space improvement in Section \ref{section:chf}.
 A process-model and the \spet{} is defined in Section \ref{sec:schedules}. 
Then the computation of a standard form as a preparation of space optimization is given in Section \ref{sec:stdform}. 
The optimization algorithm {\SPOPTn} is defined in Section \ref{sec:optn}, where also the correctness and complexity are determined
in Theorem  \ref{theorem:n-processes-polynomial}. Extensions for synchronization constructs are in Section \ref{sec:synchro}.
Section \ref{sec:applications} illustrates a relation to other scheduling methods and reports on an implementation and use 
of the algorithm. The paper concludes with Section \ref{sec:conclusion}.

\section{The Process Calculus $\CHF$ and Space Improvements}\label{section:chf}
In this section we present sufficient information to understand the role of the space optimization method of parallel processes 
in the next sections for CHF as our example programming language. Therefore 
 we first give an informal presentation of the concurrent program calculus $\CHF$ that combines distributed processes, synchronization, and shared memory 
 with a purely functional expression language. 
 We will also informally explain how our space optimization can contribute to the space behavior of program transformations 
 (so-called space improvements)  in the calculus $\CHF$. 

\subsection{The Process Calculus $\CHF$}\label{subsec:chf}

$\CHF$   models a core language of Concurrent Haskell extended by futures, 
where the exact syntax, contexts, structural congruence rules, and  reductions rules can be found for example 
in \cite{sabel-schauss-ppdp11:2011,schmidt-schauss-sabel:frank-44:11}.

\noindent A CHF-program as well as the program state after some reductions can always be represented by 
$$\MTHREAD{x}{e} \PAR \THREAD{x_1}{e_1}\PAR \ldots\PAR \THREAD{x_n}{e_n} \PAR \MVAR{y_1}{e_1'}\PAR \ldots\PAR \MVAR{y_m}{e_m'}
   \PAR z_1 = e_1'' \PAR \ldots  \PAR z_k = e_k''$$

A thread $\THREAD{x_i}{e_i}$ is a sequentially executed process, where $e_i$ is the thread-program,which finally binds its return value to $x_i$.
 The main thread $\MTHREAD{x}{\ldots}$ is a thread, with the special task to signal whether the whole computation is finished.  
 $\MVAR{y_i}{e_i'}$ is a storage device that behaves as a one-place  buffer, and $z_i = e_i''$ are shared memory cells containing the expression $e_i''$.
The expressions $e_i$, $e_i'$ and $e_i''$ are CHF-expressions, i.e. they are  monadic expressions 
(sequential and side-effecting) which may contain 
 pure expressions as in Haskell as subexpressions. The difference between $\THREAD{x_i}{e_i}$ and $x_i = e_i$ is that $\THREAD{x_i}{e_i}$ will execute,
 whereas $x_i = e_i$ is like a pointer for sharing the expression $e_i$.
 
 The execution is defined through a standard reduction sequence on the syntactic description of the program (the state), which is a non-deterministic 
 small-step reduction, where the non-determinism comes only from the competing processes.
Every thread $\THREAD{x_i}{e_i}$ can be seen as a process that performs (controlled by the standard reduction) the computation defined by expression $e_i$. 
The parallel combination of the threads performs a distributed evaluation, where also new threads may be started.

As an example of a CHF-program consider the following definition of a program, where we use the do-notation as in Haskell with the same
meaning in CHF. 
$$
   {\begin{array}{lll} \MTHREAD{x}{} & \tdo & z_1  \leftarrow (\tfuture~ e_1) \\
                                 &     & z_2  \leftarrow (\tfuture~ e_2) \\
                                  &    & \tseq~(z_1+z_2)~(\treturn (z_1,z_2))
   \end{array}
              }  
$$ 

After two reductions of the main thread, the 
  state is
$$
   {\begin{array}{lll} \MTHREAD{x}{} & \tdo &   (\tseq~(z_1+z_2)~(\treturn (z_1,z_2)))\\
                        \PAR 
          \THREAD{z_1}{e_1} \\
            \PAR 
          \THREAD{z_2}{e_2} 
   \end{array}
              }  
$$ 
which consists of three  threads. The main thread now has to wait for the delivery of the values for $z_1, z_2$, which will be the result
after the threads for $z_1, z_2$ terminate their computation and return something.  

If the expressions $e_1, e_2$ use common variables, for example if $e_2$ demands the value of $z_1$, then  the processes 
are not  independent, and the sequence of executions is restricted. There may even be deadlocks, if $e_1$ requires $z_2$ as a value, and 
$e_2$ requires $z_1$ as a value. 

In the case that the expressions $e_1, e_2$ do not use common variables (even not indirectly), the processes can be evaluated independently, which means that
every interleaving of the executions of $e_1, e_2$ is possible. This independent case will be considered more deeply in later sections,
since it permits nice space optimizations, and an example for easy detection of space improvements.

\begin{figure*}[tp]
\begin{center}
 $\begin{array}{l@{~}c@{~}lp{6cm}}
  \size(x) &= & 0\\
  \size(e_1~e_2) &= & 1+ \size(e_1) + \size(e_2)\\
  \size(\lambda x.e) & = & 1+ \size(e)\\
  \size(\tcase~e~\tof~\talt_1 \ldots \talt_n) & = & 1+ \size(e)   ~~+ \sum_{i=1}^n \size(\talt_i)\\
  \size((c~x_1 \ldots x_n)~\casepf~ e) & = & 1 + \size(e) \\
    \size(f~e_1 \ldots e_n) & = & 1 + \sum \size(e_i) & \hspace*{-1.5cm}\mbox{for constructors and operators $f$}\\
                  &&   & \hspace*{-1cm}\mbox{such as \tfuture,\treturn, \ldots  } \\
  \size(\tletr~x_1 = e_1, \ldots,  x_n = e_n~\tin~s) & = & \size(e)  + \sum \size(e_i)\\
  \size(P_1 \PAR P_2) & = & \size(P_1)  + \size(P_2)\\
  \size(x ~\mathit{op}~ {e}) & = & 1 +  \size(e) & \hspace*{-1.5cm}\mbox{for $\mathit{op}\in\{\mathrm{=,\Leftarrow,\mathbf{m}}\}$}\\
   \size(\EMPTYMVAR{x} ) & = & 1 \\
    \size(\NEW{x}.{P} ) & = &  1+\size(P)  
  \end{array}$
\end{center}

  \caption{Definition of $\size{}$ of expressions}\label{fig:space-size}
\end{figure*}

In addition to the program executions, it is crucial
to recognize  (binding-)garbage and remove it, since we are interested in space improving transformations. 
 It is shown in \cite{schmidt-schauss-sabel-dallmeyer:frank-58:17,schmidt-schauss-sabel-dallmeyer-ppdp-18} that garbage collection and the 
 modification of the standard reduction (i.e. program execution) leaves all interesting properties
(equivalence of expressions, correctness of transformations) invariant, and thus this is a correct and space-optimizing transformation.

\subsection{Space Measure, Equivalence of Programs and Space Improvements} 
An example for a  space measure is the  generalization of the space measure of \cite{schmidt-schauss-dallmeyer-wpte:17,schmidt-schauss-sabel-dallmeyer:frank-58:17}, 
which does not count variables (see Fig. \ref{fig:space-size}). The reason for the specifics is that this size measure is compatible with 
the variants of abstract machines for CHF as explained in \cite{schmidt-schauss-dallmeyer-wpte:17,schmidt-schauss-sabel-dallmeyer:frank-58:17}. 

\begin{definition}\label{def:of-spmin}
The space measure $\spm(\mathit{Red})$ of a successful standard reduction $\mathit{Red}$ of a program  $P$ is the maximum of all sizes $\size(P_i)$ during
the whole standard reduction  sequence, $\mathit{Red} = P \xrightarrow{sr}  P_1 \xrightarrow{sr} \ldots \xrightarrow{sr} P_n$,
   where we assume that the $P_i$ are always garbage-reduced. 

The space measure of a CHF-program  $P$ is defined as $\spmin(P) = min\{\spm(\mathit{Red}) \mid \mathit{Red}$ 
~ is a successful standard reduction of $P\}$.
\end{definition}
 
As a concrete example, the size of the program  $(\MTHREAD{x}{\treturn~y} \PAR  \MVAR{x}{1}  \PAR y = {\tt Cons}~x~{\tt Nil})$ is $2+2+2 = 6$. 

The reason for not counting the sizes directly before a garbage collection is that the calculus and abstract machines may create bindings
 that may be garbage 
and would thus be immediately garbage collected after the reduction step.
Taking this garbage into account  would distort  the reasoning about measurement  in particular  if these bindings have a large size 
    (more information about this can be found in \cite{schmidt-schauss-dallmeyer-wpte:17}).
    This principle  of measuring space in a small-step calculus is also used in \cite{gustavsson-sands:01}. 
    
    In the following $P\maycon$ means that $P$ has a successful standard reduction and $P\maydiv$ is its negation; $P_1\sim_c P_2$ means that $P_1, P_2$ 
    are contextually equivalent in $\CHF$. 

\begin{definition}\label{def:space-improvement}
  A program transformation $\xrightarrow{\PT}$ is a {\em space-improvement} if for all contextual equivalent processes $P_1,P_2$: $P_2 \xrightarrow{\PT} P_1$ 
  implies that $P_1$ space-improves $P_2$, i.e. $\spmin(P_1) \leq \spmin(P_2)$.
\end{definition}

In this paper we focus on a special situation, where the program $P$ consists of several threads that, after they are started, 
run completely independent, without using common data structures, and then communicate and halt.  In order to test or prove $P \longrightarrow P'$ to be an space-improvement, 
it is crucial to determine
the optimal space usage of $P$ and compare it with the optimal space usage of $P'$. 
The computation of the optimal space usage of $P$ requires (among others) to find the space-optimal interleaving of the phase between 
starting the $n$ threads until all threads finally stop and communicate. 

For example, in the program $\MTHREAD{u}{\ldots} \PAR \ldots \PAR \THREAD{x}{e_1} \PAR \THREAD{y}{e_2}$, we consider $\THREAD{x}{e_1}$ and $\THREAD{x}{e_2}$ 
as the two subprocesses $p_1,p_2$, which can be measured separately.

For the case of independent processes, we present an algorithm for computing an optimal interleaving and the space minimum in the case of parallel evaluation possibilities (if the executions are already given),
where the algorithm runs in $O(n\log n)$ time.
We also analyze the impact on the runtime in the case of dependencies between processes, where synchronization points between processes are defined explicitly.

\section{Abstract Model of Independent Processes and Space}\label{sec:schedules}

The assumptions underlying the abstraction is that CHF-processes use a common memory 
for their local data structures, but they cannot see each others memory entries.
The CHF-processes may independently start or stop or  pause  at certain time points. We also assume
 that synchronization and communication 
  may occur at certain time points as interaction between CHF-processes. \\ 
Every CHF-process is abstractly modeled  by its trace of space usage, given as a list of integers.
In addition we later add constraints expressing simultaneous occurrences of time points of different CHF-processes 
 as well as start-points and end-points of CHF-processes. 

In the following  we use the notation $[a_1,\ldots,a_n]$ for a list of the elements $a_1,\ldots,a_n$. We also use 
 $(a \cons l)$ for adding a first element $a$ to list $l$, $l_1 \append l_2$ for appending the lists $l_1$ and $l_2$, 
 $\mathit{tail} (l)$ for the tail of the list $l$,  
and $[f(x) \mid x \in L]$ for a list $L$ denotes the list of $f(x)$ in the same sequence as that of $L$ (i.e. it is a list comprehension).

In the following we abstract CHF-processes by a list of non-negative integers. 
For simplicity we call this list a {\em process} in the rest of the paper.
A (parallel) interleaving is constructed such that from one state to the next one, each process proceeds by at most one step 
and at least one process  proceeds.

\begin{definition}\label{def:integer-process}
A {\em process} is a nonempty, finite list of non-negative integers.
For $n > 0$ let $P_1,\dots,P_n$ be  processes where $m_i$ is the length of $p_i$, and let $p_{i,j}$ for $j = 1,\ldots, m_i$ be the elements. 
Then   an {\em interleaving} of $P_1,\ldots,P_n$ is a list $[q_1,\ldots,q_h]$ of
 $n$-tuples $q_j$ constructed using the following (non-deterministic) algorithm:
 \begin{enumerate}
   \item Initially, let $q$ be the empty list.
   \item\label{item-restart} If all processes $P_1,\dots,P_n$ are empty, then  return $q$.
   \item   Set $q := q \append  [(p_{1,1}, \ldots, p_{n,1})]$, i.e., the tuple of all first elements is added at the end of $q$.\\
    Let $(b_1,\ldots,b_n)$ be a (nondeterministically chosen) tuple of Booleans, such that
      there is at least one $k$ such that $b_k$ is {\tTrue} and $P_k$ not empty.  \\
    For all $i = 1,\ldots,n$: set $P_i = \mathit{tail}(P_i)$ if $b_i$ and $p_i$ is not empty; otherwise do not change $P_i$. \\
    Continue with item  \ref{item-restart}.
\end{enumerate}
\end{definition}

\begin{definition} Let $P_1,\ldots,P_n$ be processes. 
The {\em space usage} $\spm(S)$ of an interleaving $S$ of $P_1,\ldots,P_n$ is the maximum of the sums of the elements in the tuples in $S$, i.e. 
$\spm(S) = \max\{\sum_{i = 1}^n a_i \mid (a_1,\ldots,a_n) \in S \}$.
The {\em required space} $\spmin(P_1,\ldots,P_n)$ for $n$ processes $P_1,\ldots,P_n$  is the minimum of 
the space usages of all interleavings of $P_1,\ldots,P_n$, i.e. $\min\{\spm(S) \mid S \mbox{ is an interleaving of } P_1,\ldots,P_n\}$.\\
A {\em peak} of $P_i$ is a maximal element of $P_i$, and a {\em valley} is a smallest element in $P_i$.
A {\em local peak} of $P_i$ is an maximal element in $P_i$ which is not smaller than its neighbors.
A {\em local valley} of $P_i$ is a minimal element in $P_i$ which is not greater than its neighbors.
\end{definition}

\begin{example}
 For two processes $[1,7,3], [2,10,4]$ the $\spmin$-value is $11$, by first running the second one and then running the first.
 I.e. such a (space-optimal) interleaving is $[(1,2),$ $(1,10),$ $(1,4),$ $(7,4),$ $(3,4)]$.
The interleaving that results from an ``eager'' scheduling is $[(1,2), (7,10), (3,4)]$, with $\spm$-value 17, and hence is not space-optimal.
\end{example}

\section{Standard Form of Processes}\label{sec:stdform}
We will argue that  an iterated reduction of single processes by the following 5  patterns  permits to compute $\spmin$ from smaller processes.
This is a first step like a standardization of processes for the purpose of $\spmin$-computation, and is a preparing step for the 
optimization algorithm {\SPOPTn} in Definition \ref{def:algo-n-processes}.

\begin{definition}  The trivial pattern $M_0$ is $a_i = a_{i+1}$.
There are two further, nontrivial patterns: The first pattern $M_1$ is  $a_i \leq a_{i+1} \leq a_{i+2}$ and the second pattern  $M_2$ is
 $a_i \geq a_{i+1} \geq a_{i+2}$. \\
 A pattern {\em matches a process}  $[a_1,\ldots,a_k]$  at index $i$, if for index $i$ the conditions are satisfied.\\ 
 A single pattern application is as follows: If the patterns $M_0, M_1$ or $M_2$ matches a process for some index $i$, then  $a_{i+1}$ is removed.
 \end{definition}

\begin{proposition}
 Let $P_1,\dots, P_n$ be $n$ processes and let $P_1', \dots, P_n'$ be the processes after removal of subsequent equal entries, i.e. using $M_0$.
  Then $\spmin(P_1,\dots,P_n) =  \spmin(P_1',\dots,P_n')$.  
\end{proposition} 
\begin{proof}
This is obvious by rearranging the schedules, leading to different \spets{}, which have the same $\spmin$-value.
\end{proof}

\begin{proposition}
Let $P_1,\dots,P_n$ be $n$ processes. Let $P_1',\dots,P_n'$ be the processes after several application of the pattern-reduction process using $M_1$ 
and $M_2$. 
Then $\spmin(P_1,\dots,P_n) = \spmin(P_1',\dots,P_n')$.
\end{proposition}
\begin{proof}
It is sufficient to assume that exactly one change due to a pattern match is performed. It is also sufficient to assume that the pattern is $M_1$ and that it applies in $P_1$.
We can also look only at a subpart of an \spet{} to have easier to grasp indices.
For argumentation purposes, we choose the correspondence between the \spets{} $(P_1,P_2,\ldots,P_n)$ and $(P_1',P_2,\ldots,P_n)$ as follows.  \\
Let $[p_{1,1},p_{1,2},p_{1,3}]$ with $p_{1,1} \leq p_{1,2} \leq p_{1,3}$ 
be the subprocess of $P_1$ that is replaced by $[p_{1,1},p_{1,3}]$. Consider the part
 $(p_{1,1},\dots,p_{n,1})\cons [(p_{1,2},p_{2,2},\dots,p_{n,2}) \mid (p_{2,2},\dots,p_{n,2}) \in B] \append [(p_{1,3},\dots,p_{n,3})]$
of the \spet{}, where $B$ is a sequence of $n-1$-tuples.
Then the modified {\spet} for $(P_1',P_2,\dots,P_n)$ can be defined as:
 ~~  $(p_{1,1},\dots,p_{n,1})\cons [(p_{1,1},p_{2,2},\dots,p_{n,2}) \mid (p_{2,2},\dots,p_{n,2}) \in B] \append [(p_{1,3},\dots,p_{n,3})]$\\
  and since for every \spet{} of $(P_1,P_2,\ldots,P_n)$ we obtain an \spet{} of  $(P_1',P_2,\ldots,P_n)$ with a $\spm$ that is smaller or equal, and since $\spmin$ is defined
  as a minimum, we obtain
  $\spmin(P_1,P_2,\ldots,P_n) \geq  \spmin(P_1',P_2,\ldots,P_n)$.   \\
For the other direction, consider the part $[(p_{1,1},\dots,p_{n,1}),(p_{1,3},p_{2,2},\dots,p_{n,2})]$ of an \spet{} of the processes $P_1',P_2,\dots,P_n$. 
Then $\spmin$ of the part $[(p_{1,1},\dots,p_{n,1}), (p_{1,2},p_{2,2},\dots,p_{n,2}), (p_{1,3},p_{2,2},\dots,p_{n,2})]$ of the \spet{} of $P_1,\dots,P_n$ is the same as
 before, thus $\spmin(P_1,\dots,P_n) \leq \spmin(P_1',P_2,\dots,P_n)$. \\
The two inequations imply $\spmin(P_1,\dots,P_n) = \spmin(P_1',P_2,\dots,P_n)$. 
\end{proof}

\begin{definition} If in a process $P$ 
   every strict increase is followed by a strict decrease and every strict decrease is followed by a strict increase,
then the process $P$ is called a {\em zig-zag process}.
\end{definition}

By exhaustive application we can assume that the pattern $M_0$, $M_1$ and $M_2$  above are not applicable to processes which means that the processes can be assumed to be zig-zag.

Now we show that there are more complex patterns that can also be used to reduce the processes before computing $\spmin$. 
The following patterns $M_3, M_4$ are like stepping downstairs and upstairs, respectively.

\begin{definition}
  The patterns $M_3, M_4$ are defined as follows: 
  \begin{itemize}
\item  $M_3$ consists of $a_i,a_{i+1}, a_{i+2},a_{i+3}$, with $a_i > a_{i+1}$, $a_{i+1} < a_{i+2}$, $a_{i+2} > a_{i+3}$ and $a_i \geq a_{i+2},a_{i+1} \geq a_{i+3}$. 
\item  $M_4$ consists of $a_i,a_{i+1}, a_{i+2},a_{i+3}$, with $a_i < a_{i+1}$, $a_{i+1} > a_{i+2}$, $a_{i+2} < a_{i+3}$ and $a_i \leq a_{i+2},a_{i+1} \leq a_{i+3}$. 
\end{itemize}
 
$\begin{array}{@{}lc@{\hspace*{1.5cm}}lc@{}}
M_3:
&
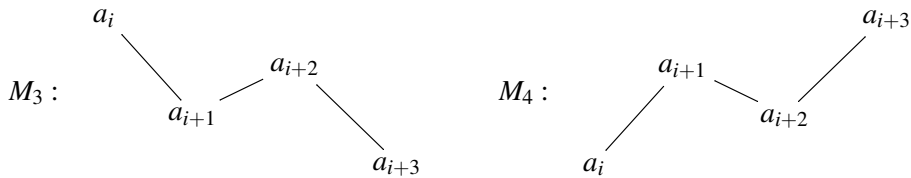
\begin{minipage}{4cm}
$\xymatrix@C=5mm@R=2mm{
  a_i\ar@{-}[ddr]   \\
   & & a_{i+2} \ar@{-}[dl] \ar@{-}[ddr] \\
    & a_{i+1}  \\
    & & & a_{i+3} \\
}$
\end{minipage}
 &
 M_4: &
 \begin{minipage}{4cm}
$\xymatrix@C=5mm@R=2mm{
   & & & a_{i+3} \ar@{-}[ddl] \\
   & a_{i+1}  \ar@{-}[ddl] \ar@{-}[dr]  &   \\
   && a_{i+2}   \\
   a_i \\
}$
\end{minipage}
 \end{array}
$

 If for some $i$:  $M_3$ or $M_4$ matches, then eliminate  $a_{i+1}, a_{i+2}$.
\end{definition}

We show that the complex patterns can be used to restrict the search for an optimum to special processes:
\begin{lemma}
 Let $P_1,\ldots, P_n$ be processes.
 If one of the patterns $M_3,M_4$ matches one of the processes, then 
 it is sufficient to check the shortened $P_1',\ldots, P_n'$ for the space-minimum.
\end{lemma}

\begin{proof}
It is sufficient to assume that exactly one change due to a pattern match is performed. It is sufficient to assume that the pattern is $M_3$ and 
that it applies in $P_1$.
We can also look only at a subpart of an \spet{} to have easier to grasp indices.
For argumentation purposes, we choose the correspondence between the \spets{} $(P_1,P_2,\ldots,P_n)$ and $(P_1',P_2,\ldots,P_n)$ as follows.\\
Let $[p_{1,1},p_{1,2},p_{1,3},p_{1,4}]$ with $p_{1,1} > p_{1,2}$, $p_{1,2} < p_{1,3}$, $p_{1,3} > p_{1,4}$, $p_{1,1} \geq p_{1,3}$ and $p_{1,2}\geq p_{1,4}$
be the subprocess of $P_1$ that is replaced by $[p_{1,1},p_{1,4}]$. Consider the following part of the \spet{}, where $B_2,B_3$ are sequences of $n-1$-tuples:\\
$\begin{array}[t]{l}
(p_{1,1},\dots,p_{n,1})\cons [(p_{1,2},p_{2,2},\dots,p_{n,2}) \mid (p_{2,2}\dots,p_{n,2}) \in B_2] \append [(p_{1,3},p_{2,3},\dots,p_{n,3}) \mid (p_{2,3}\dots,p_{n,3}) \in B_3]\\[1mm]
\quad \append [(p_{1,4},\dots,p_{n,4})]
\end{array}$
\\[1mm]
The modified interleaving for $(P_1',P_2,\dots,P_n)$ can be defined as
  $[(p_{1,1},\dots,p_{n,1})] \append [(p_{1,4},q_{2,4},\dots,q_{n,4}) \mid (q_{2,4},\dots,q_{n,4}) \in B_2\append B_3 \append [p_{2,4},\dots,p_{n,4}]]$
and since for every \spet{} of $(P_1,P_2,\ldots,P_n)$ we obtain an \spet{} of  $(P_1',P_2,\ldots,P_n)$ with a $\spm$ that 
is smaller or equal, and since $\spmin$ is defined as a minimum, we obtain $\spmin(P_1,P_2,\ldots,P_n) \geq  \spmin(P_1',P_2,\ldots,P_n)$.   \\
Now consider the part $[(p_{1,1},\dots,p_{n,1}),(p_{1,4},\dots,p_{n,4})]$ of an \spet{} of $P_1',P_2,\ldots,P_n$. 
Then $\spmin(.)$ of the part 
$[(p_{1,1},p_{2,1},\dots,p_{n,1}),(p_{1,2},p_{2,1}\dots,p_{n,1}),(p_{1,3},p_{2,1},\dots,p_{n,1}),(p_{1,4},\dots,p_{n,4})]$
 of the \spet{} of $P_1,\dots,P_n$ is the same as
 before, thus $\spmin(P_1,\dots,P_n) \leq \spmin(P_1',P_2,\dots,P_n)$. 

The two inequations imply $\spmin(P_1,\dots,P_n) = \spmin(P_1',P_2,\dots,P_n)$. 
\end{proof}

\begin{definition}
    A process  $[a_1,b_1,a_2,b_2,\ldots,a_n]$  (or $[b_0,a_1,b_1,a_2,\ldots,a_n]$, or $[a_1,b_1,$ $a_2,b_2,$ $\ldots,a_n,b_n]$, 
    or $[b_0,a_1,b_1,a_2,b_2,\ldots,a_n,b_n]$, resp.) 
    is a 
  monotonic increasing zig-zag (mizz), iff
   $a_i < b_j$ for all $i,j$, and 
   $a_1,a_2,\ldots,a_n$ is strictly monotonic decreasing, and  
    $b_1,b_2,\ldots,b_{n-1}$ (and $b_0,b_1,b_2,\ldots,b_{n-1}$ and $b_0,b_1,b_2,$ $\ldots,b_{n-1},b_n$, resp.) is strictly monotonic increasing.  
    
 A process  $[a_1,b_1,\ldots,a_n]$ 
   is a 
  monotonic-decreasing zig-zag (mdzz), iff
   $a_i < b_j$ holds for all $i,j$, and 
   $a_1,a_2,\ldots a_n$ is strictly monotonic increasing, and  
    $b_1,b_2,\ldots, b_{n-1}$ (or  $b_0,b_1,b_2,\ldots, b_{n-1}$, resp. ) is strictly monotonic decreasing.  

    A process is {\em midzz}, if it is a mizz followed by a mdzz. More rigorously, there are essentially two cases, 
       where we omit the cases with end-peaks and/or start-peaks.
  \begin{enumerate}
    \item the mizz  $[a_1,b_1,a_2,b_2,\ldots,a_n]$ and the mdzz $[a_1',b_1',\ldots,a_n']$, where $a_n = a_1'$ are combined to
     $[a_1,b_1,a_2,$ $b_2,\ldots,a_n,b_1',\ldots,a_n']$,
  \item the mizz  $[a_1,b_1,a_2,b_2,\ldots,a_n,b_n]$ and the mdzz $[b_0',a_1',b_1',\ldots,a_n']$, where $b_n = b_0'$ are combined to
     $[a_1,b_1,a_2,b_2,\ldots,a_n,b_n,a_1',b_1',\ldots,a_n']$. 
  \end{enumerate} 
\end{definition}
\noindent Typical graphical representations of  mizz- and mdzz-sequences are:\\
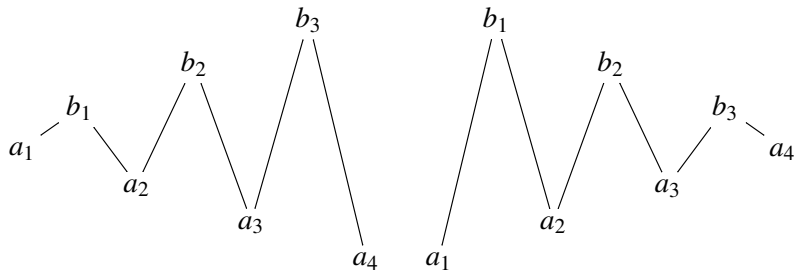
 \begin{tabular}{l@{\qquad\qquad}l}
\begin{minipage}{4cm}  
$  \xymatrix@C=2mm@R=0.5mm{
   &&&&&b_3\ar@{-}[ddddddr] \\
  &&&b_2\ar@{-}[ddddr] \\
   & b_1 \ar@{-}[ddr] \\
  a_1 \ar@{-}[ur] \\
  &&a_2\ar@{-}[uuur] \\
  &&&&a_3\ar@{-}[uuuuur] \\
  &&&&&&a_4\\
 }
$
\end{minipage}

&
 \begin{minipage}{4cm} 
 $  \xymatrix@C=2mm@R=0.5mm{
  &b_1 \ar@{-}[dddddr] \\
   &&& b_2 \ar@{-}[ddddl] \ar@{-}[dddr] \\
   && & && b_3 \ar@{-}[ddl] \ar@{-}[dr] \\
     &  && & &&  a_4\\
     &&& &  a_3\\
     && a_2\\
 a_1 \ar@{-}[uuuuuur] \\
 }
$
 \end{minipage} 
\end{tabular}

If the goal is to compute the optimal space, then there are several reduction operations on processes that 
ease the computation and help us to concentrate on the hard case. 
First we show that one-element processes can be excluded, and second that processes with start- or end-peaks can be reduced by omitting 
elements. 
Then we show that through the use of the $5$ patterns $M_0, \ldots, M_4$ for reductions we can concentrate on special forms of zig-zag-processes, so-called midzz. 
 
\begin{proposition}\label{prop:reduce-1-element-process}
If $P_1 = [a_1]$ and $P_2, \ldots,P_n$ are processes then $\spmin(P_1,\ldots,P_n) = a_1+\spmin(P_2,\ldots,P_n)$.
\end{proposition}
\begin{proof} $a_1$ is the first element of every tuple in any {\spet} of $P_1,\ldots,P_n$, hence  the claim is valid.
\end{proof}

\begin{proposition}\label{prop:start-and-end-peak}
Let $P_i = [p_{i,1},\dots,p_{i,n_i}]$ for $i = 1,\ldots,n$ be processes. If $p_{1,1}$ is a start-peak of $P_1$, then let $P_1' = [p_{1,2},\dots,p_{1,n_1}]$.
 Then  
$\spmin(P_1,\dots,P_n)=\max(\sum_{i} p_{i,1},\spmin(P_1',P_2,\dots,P_n))$. The same holds symmetrically if $P_1$ ends with a local peak.
\end{proposition}
\begin{proof}
   Let $q=[(p_{1,1},q_{1,2},\dots,q_{1,n}),\ldots,(p_{1,1},q_{h,2},\dots,q_{h,n})] \append$  $[(p_{1,2},q_{h+1,2},\dots,q_{h+1,n})] \append R$  be an \spet{} for $P_1,\dots,P_n$ and 
   some $h$. 
   If $h \not= 1$,  this can be changed to 
   $[(p_{1,1},q_{1,2},\dots,q_{1,n}),\!(p_{1,2},q_{2,2},\dots,q_{2,n}),$ $ \ldots,$  $ (p_{1,2},q_{h,2},\dots,q_{h,n})]$  $\append$  $[(p_{1,2},q_{h+1,2},\dots,q_{h+1,n})] \append R$ 
     without increasing the necessary space. 
   Hence $\spmin(P_1,\dots,P_n)$ $\ge$  $\max(\sum_ip_{i,1},$ \mbox{$\spmin(P_1',P_2,\dots,P_n))$}. 
   
   On the other hand, if we  have a space-optimal schedule of $P_1',P_2,\dots,P_n$, then we can extend this by starting with $(p_{1,1},\dots,p_{n,1})$ and obtain $\spmin(P_1,\dots,P_n)\le \max(\sum_ip_{i,1},\spmin(P_1',P_2,\dots,P_n))$.

   Hence $\spmin(P_1,\dots,P_n) = \max(\sum_ip_{i,1},\spmin(P_1',P_2,\dots,P_n))$.
\end{proof}   

\begin{lemma}
We can assume that processes $P_1,\dots,P_n$ are all of length at least 3 for computing the optimal space.
\end{lemma}
\begin{proof}
Proposition \ref{prop:reduce-1-element-process} permits to assume that the length is at least $2$. 
Proposition \ref{prop:start-and-end-peak}
allows to assume that there is no start- nor an end-peak. 
Hence we can assume that  processes are of length at least $3$.
\end{proof}

\begin{lemma}
 Let $P$ be a process that  starts and ends with local valleys.
 Then the application of the patterns $M_0, \ldots, M_4$
 with subsequent reduction always produces a process that also starts and ends with local valleys.
 \end{lemma}
 \begin{proof}
 The reduction either removes according to pattern $M_0$ or it removes inner entries of the lists.
 \end{proof}

\begin{proposition}\label{prop:midzz-cases}
 A process such that none of the  patterns $M_0,$ $M_1,$  $M_2,$ $M_3,$  $M_4$  matches and which does not start or end with a local peak is a midzz. 
 \end{proposition}

 \begin{proof} 
 
We consider all four different cases how small sequences may proceed, if no pattern applies.
\begin{longtable}[t]{lll}
1.
&
$\xymatrix@C=3mm@R=0.5mm{
   a_1\ar@{-}[dddr]   & &  \\
  & & a_3 \ar@{-}[ddl]   \\
    &&&    a_4?  \\
  & a_2 
 }$
&
\begin{minipage}[t]{0.7\textwidth}
Case $a_1 > a_2$, $a_2 < a_3$  and $a_3 < a_1$. Then $a_4 < a_3$. The relation  $a_4 \leq a_2$ is not possible, since then pattern 
$M_3$ matches. Hence $a_3 > a_4 > a_2$. Then $a_1,a_2,a_3,a_4$ is a tail of a mdzz. \\
The case $a_1 > a_2$, $a_2 < a_3$  and $a_3 = a_1$ leads to the same relations $a_3 > a_4 > a_2$.  
Then $a_2,a_3,a_4$ is a tail of a mdzz.
\end{minipage}

\\\\[-1ex]

2.
&
$\xymatrix@C=3mm@R=0.5mm{
& a_2 \\
 &&&    a_4?  \\
& & a_3 \ar@{-}[uul]   \\
 a_1\ar@{-}[uuur]   & &  \\
}$
&
\begin{minipage}[t]{0.7\textwidth}
Case $a_1 < a_2$, $a_2 > a_3$  and $a_3 > a_1$. Then $a_4 > a_3$. The relation  $a_4 \geq a_2$ is not possible, since then pattern $M_4$ matches. Hence $ a_3 < a_4 < a_2$. Then $a_1,a_2,a_3,a_4$ is a mdzz. \\
The case $a_1 < a_2$, $a_2 > a_3$  and $a_3 = a_1$ leads to the same relations $a_3 < a_4 < a_2$. Then using case 1 for the the next element $a_5$, the sequence $a_3,a_4,a_5$ is a tail of a mdzz.
\end{minipage}
\\\\[-1ex]

3.
&
$\xymatrix@C=3mm@R=0.5mm{
&& a_3\ar@{-}[ddl]   \\
a_1 \ar@{-}[dr] & & &  a_4 ?  \\
& a_2 
}$
&
\begin{minipage}[t]{0.7\textwidth}
Case $a_1 > a_2$, $a_2 < a_3$ and $a_3 > a_1$. Then $a_3 > a_4$ and there are three cases:\\[-0.5cm]
\begin{itemize}
\setlength{\itemsep}{0pt}
\item[(i)] If $a_4 = a_2$ then the sequence starting from $a_3$ is a mdzz.
\item[(ii)] If $a_4 > a_2$ then case 2 is applicable and  the sequence starting from $a_2$  is a mdzz.
\item[(iii)] If $a_4 < a_2$ then the  sequence starting with $a_1$ proceeds  as mizz. It may later turn into a mdzz.   
\end{itemize}
\end{minipage}

\\\\[-1ex]

4.
&
$\xymatrix@C=3mm@R=01.5mm{
& a_2 \\
a_1 \ar@{-}[ur] & & &  a_4 ?  \\
&& a_3\ar@{-}[uul]   \\
}$
&
\begin{minipage}[t]{0.7\textwidth}
Case $a_1 < a_2$, $a_2 > a_3$ and $a_3 < a_1$. Then $a_3 < a_4$ and there are three cases:\\[-0.5cm]
\begin{itemize}
\setlength{\itemsep}{0pt}
\item[(i)] If $a_4 = a_2$ then  the sequence starting from $a_3$ is a mdzz.
\item[(ii)] If $a_4 < a_2$ then case 1 is applicable and the sequence starting from $a_2$ is a mdzz.
\item[(iii)] If $a_4 > a_2$ then the  sequence starting with $a_1$ proceeds as mizz. It may later turn into a mdzz.
\end{itemize}
\end{minipage}
\end{longtable}
Now we put the parts together and conclude that the sequence must be a midzz. \qedhere
\end{proof}

 Note  that the definition of midzz permits the simplified case that the process is a mizz or mdzz.

\begin{definition}
A   process is called {\em standardized} if it is a midzz of length at least 3, and does not start nor end with a local peak.
\end{definition} 
 
 \begin{lemma}Let $P$ be a midzz-process, where no pattern $M_0,$ $M_1,$ $M_2,$ $M_3, M_4$ applies,  and which is of length at least $3$, 
 and does not start nor end with a local peak:
   Then a midzz-process has one or two global peaks, it has one or two global valleys, but not two global peaks and two global valleys at the same time.
 \end{lemma}

 \begin{proof}
 The considerations and cases in the proof of Proposition \ref{prop:midzz-cases} already exhibit the possible cases. \\
 Since the patterns $M_3, M_4$  do not apply, there cannot be three global peaks nor three global valleys.
 If there are two global peaks and two global valleys, then the picture is  
   
  $  
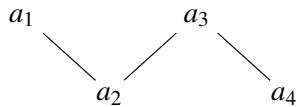
\xymatrix@C=6mm@R=2mm{
  a_1 \ar@{-}[ddr] && a_3  \ar@{-}[ddl] \ar@{-}[ddr]  \\
  & & \\
   & a_2 & & a_4   
 }
$

and we can apply pattern $M_3$, which is forbidden by the assumptions. Similarly for the case where $a_1$ is a global valley.  
 \end{proof}          

Hence, a standardized process in midzz-form has three different possibilities for the global peaks and valleys: 
(i) there is a unique global peak and a unique global valley; (ii) 
there is a unique global peak and two  global valleys;  (iii) there are two global peaks and a unique global valley.

\section{Optimizing Many Independent Processes}\label{sec:optn}
 
 Let us assume in this section that there are $N$ processes $P_1,\ldots,P_N$ of total size $n$. 
 \begin{algorithm}[Standardization]\label{alg:standardization}
 For an input of $N$ processes $P_1,\ldots,P_N$:  
 \begin{enumerate}
   \item For every process $P_i$ in turn: Scan $P_i$ by iterating $j$ from $0$ as follows:\\
     If the patterns $M_0,\ldots,M_4$ apply at index $j$ then reduce accordingly and restart the scan at position $j-3$, otherwise go on with index $j+1$.

    \item Let $K_0$ be the sum of all first elements of $P_1,\ldots,P_N$.
          Let $P'_1 ,\ldots,P'_N$ be obtained from $P_1,\ldots,P_N$ by removing all start-peaks only from processes of length at least 2.
     \item Let  $K_\omega$ be the sum of all last elements of $P_1',\ldots,P_N'$.
         Let $P''_1 ,\ldots,P''_N$ be obtained from $P'_1,\ldots,P'_N$ by removing all  end-peaks only from processes of length at least 2.
    \item  Let $A$ be the sum of all elements of one-element processes, and 
        let $P'''_1 ,\ldots,P'''_{N'}$ be   $P''_1 ,\ldots,P''_N$ after removing all one-element processes.
    \item If $M'''$  is $\spmin(P'''_1 ,\ldots,P'''_{N'})$, then 
           $\spmin(P_1,\ldots,P_N)$ is computed as $\max(M'''+A,K_0,K_\omega)$.
 \end{enumerate}
 \end{algorithm}
  
\begin{theorem}\label{thm:standardized-spmin-complexity} Algorithm \ref{alg:standardization} for standardization reduces the 
computation of $\spmin$ for $N$ processes $P_1,\ldots,P_N$ of size $n$ 
to the computation of $\spmin$ for standardized processes in time $O(n)$.
\end{theorem}
\begin{proof}
 Algorithm \ref{alg:standardization} is correct by Propositions  \ref{prop:reduce-1-element-process}
 and \ref{prop:start-and-end-peak}. \\
 The required number of steps for pattern application is $O(n)$: Every successful application of a pattern strictly reduces the number of elements. 
 The maximum number of steps back is $3$, hence at most $4n$ total steps are necessary. 
 Stepping back for $3$ is correct, since a change at index $k$ cannot affect pattern application for indices less than $k-3$.    
 The overall complexity is $O(n)$ since scans are sufficient to perform all the required steps and computations
 in Algorithm \ref{alg:standardization}.  
 \end{proof}

\begin{algorithm}\label{algorithm:space-n}{\bf Algorithm for Left-Scan of $N$ processes.}
We describe an algorithm for standardized processes which performs a left-scan until a global valley is reached and returns the required space 
for the left part. \\
The following index $I_{i,ends}$ in process $P_i$ for $i = 1,\ldots,N$ is fixed: 
It is the index in $P_i$ of the global valley, if it is unique, and of the rightmost
global valley if there are two global valleys.

\begin{enumerate}
\item Build up a search tree $T$ that contains pairs $((p_{i,2}-p_{i,1}),i)$ for each process $P_i=[p_{i,1},\dots,p_{i,n_i}]$, where the first component is the search key. % O(k*log k): O(k) for the calculation of the unsorted list and O(k*log k) to sort this list.
\item Set $S=M=\sum_i p_{i,1}$. Also for each process $P_i$ there are indices $I_i$ indicating the current valley positions of the process, initially set $I_i=1$ for each process. % O(k)
\item\label{alg-n-step3} If $T$ is empty then return $M$ and terminate. % O(1)
\item\label{alg-n-step4} Remove the minimal element $V=(d,i)$ from $T$. % O(log k)

If $I_i+2 \leq I_{i,ends}$, then set $M=\max(M,S+d)$, $S=S+(p_{i,3}-p_{i,1})$, insert $(p_{i,4}-p_{i,3},i)$ into $T$ (only if $P_i$ contains at least 4 elements), set $I_i=I_i+2$ and remove the first two elements from $P_i$. 
Note that $P_i$ is not considered anymore in the future if $I_i+2>I_{i,ends}$ or if there is no further peak in $P_i$ after $I_i$. % O(log k)

Goto \eqref{alg-n-step3}.
\end{enumerate} 

\end{algorithm}

\noindent The right-to-left algorithm is the symmetric version and yields also the required space for the right part. 

 \begin{algorithm}\label{algorithm:SPOPTn} {\SPOPTn}\  {\bf Computation of $\spmin$ for $N$ processes}\label{def:algo-n-processes} 
\begin{enumerate}
  \item Let $M_{start}$ be the sum of all start elements, and $M_{end}$ be the sum of all end elements of the given processes $P_1,\ldots,P_N$.\\ Also let $M_{one}$ be the sum of all elements of one-element-processes.
  \item Transform the set of processes into standard form.
  \item Compute $\mathit{M}_{left}$ using the left-to-right scan and $\mathit{M}_{right}$ using the right-to-left scan.
  \item Return the maximum of ($M_{left}+M_{one})$, $(M_{right}+M_{one})$, $M_{start}$ and $M_{end}$.
\end{enumerate}
  \end{algorithm}

\begin{theorem}\label{theorem:n-processes-optimal} 
Algorithm \ref{def:algo-n-processes}, {\SPOPTn},  computes $\spmin$ of $N$ processes.
\end{theorem}
\begin{proof}
Let $P_1,\ldots,P_N$ be $N$ processes. To achieve the standard forms Algorithm \ref{alg:standardization}  is applied. 
First we argue that for those processes the required space is at least the computed space by the left-to-right scan.

Consider a state $(i_1,\ldots,i_n)$ during the construction of a space-optimal \spet{} using space $M$, where every $i_j$ is not after the index of the smallest valley,
which means  $i_j \leq I_{j,ends}$.
An invariant of the state is that $p_{i_1} + \ldots + p_{i_n} \leq M$. 
We also assume as an invariant that the current state belongs to an optimal \spet{}.
If some $i_j$ is the position of a local peak, then the optimal \spet{} can be changed to $i_j+1$ such that  
the next tuple is $(i_1,\ldots,i_j+1,\ldots,i_n)$.
Repeating this argument, we can assume that $(i_1,\ldots,i_n)$ contains only indices of local valleys. 
Now consider the set $S$ of positions $j$ in the tuple, such that $i_j <  I_{j,ends}$. For at least one such index the optimal \spet{} must proceed. 
For the indices in $S$, the next index will be a local peak, so the best way is to look for the smallest peak $p_{i_j+1}$ for $j \in S$.
 If the sum of the spaces exceeds $M$ then we have a contradiction, since the \spet{} must proceed somewhere. 
Hence $M$ is at least  $\min\{p_{i_j+1} + \sum_{h\not=j} p_{i_h} \mid j=1,\ldots,n\}$. This argument also holds, if the indices $i_j$ for $j \not\in S$ are beyond $I_{j,ends}$,
since the valley at $I_{j,ends}$ is smaller. 
For a better efficiency the algorithm calculates these sums implicitly by keeping track of the sum of the current valleys,
 i.e. $\sum_{h} p_{i_h}$. Then it uses a search tree containing the space differences between the corresponding local valley and the next peak to step forward,
  i.e. to calculate $p_{i_j+1}$. 
For the right-to-left scan the same arguments hold, symmetrically where by slight asymmetry, we only scan to  the rightmost minimal valley for every process.

Thus we have two lower bounds $M_{left}$ and $M_{right}$ for the optimal \spet{}.

The only missing argument is that we can combine those two values. For processes that have a unique global minimal valley, the combination is trivial.
For the case of processes that have global minimal valleys, we glue together the left \spet{} with the reversed right \spet{}.
This is an \spet{} and it can be performed in space at most the maximum of $M_{left}$ and $M_{right}$. 
Concluding, the algorithm computes $\spmin$ for the input processes.
\end{proof}

\begin{theorem}\label{theorem:n-processes-polynomial} 
If there are $N$ processes $P_1,\ldots,P_N$ of total size $n$, then the optimal space and an optimal schedule can be computed in time 
$O(N\log N + n\log N)$.
\end{theorem}
\begin{proof}
The algorithm {\SPOPTn} computes the optimal space and an optimal schedule (see Theorem  \ref{theorem:n-processes-optimal}). 
 We estimate the required time:
 The time to produce a standardized problem is linear, which follows from Theorem \ref{thm:standardized-spmin-complexity}.
The left-to-right and the right-to-left scan can be performed in time $O(N\log N+ n\log N)$: The search tree can be initially constructed in $O(N\log N)$. 
 Since the search tree contains at most $N$ elements during the whole calculation, we need $O(n\log N)$ steps for all lookups and insertions. 
 
 Note that the bit-size of the integers of the space-sizes is not relevant, since we only use addition, subtraction, and maximum-operations on these numbers. 
\end{proof}

\begin{example} This example illustrates the computation (without the optimization using search trees) as follows:\\
Let $P_1 = [10,1,12,5,7,1]$, $P_2 = [3,11,2,10,3]$ and $P_3 = [1,2,3,4,3,2,1]$.\\
Then we first can reduce the processes as follows: $P_3$ can be reduced by pattern $M_1,M_2$ to $P_3' = [1,4,1]$. $P_2$ is already a zig-zig process, therefore no pattern applies. 
$P_1$ starts with a local peak, hence we keep in mind $14$ as the sum of the first elements and replace $P_1$ by $P_1' = [1,12,5,7,1]$. 
The next step is to apply the pattern $M_3$, which reduce it to $P_1'' = [1,12,1]$. Thus the new problem is $P_1'' = [1,12,1]$, $P_2 = [3,11,2,10,3]$, $P_3' = [1,4,1]$.\\
A short try shows that $15$ is the optimum. However, we want to demonstrate the algorithm:\\
The left scan starts with $\itMax = 5$. The peak in $P_3'$ then enforces $\itMax = 8$ and $P_3'$ is not considered anymore, since the left scan reached the final position 
in $P_3$, i.e. the rightmost global valley. The peak in $P_2'$ then enforces $\itMax = 13$ and also $P_2$ is not considered anymore, since the final position is reached. 
Finally the peak in $P_1''$ enforces $\itMax = 15$ and the left scan terminates.\\
The right scan starts with $\itMax = 5$. Then the peak in $P_3'$ enforces $\itMax = 8$, after this the peak in $P_2'$ enforces $\itMax = 12$ and 
finally the peak in $P_1''$ enforces $\itMax = 15$.\\
Hence in summary, also taking the local peak at the beginning of $P_1$ into account, the result is $15$.
\end{example}

\section{Processes with Synchronizations}\label{sec:synchro}

We indicate how to extend our model to timing and synchronization restrictions. For example, in CHF writing into a filled MVar requires the process to wait until the MVar 
is empty. There are also race-conditions, for example if several processes try to write into an empty MVar, or several processes try to read the same MVar. 
These constraints are captured by the constraints below, where the race conditions can be modeled by disjunctions. 

\begin{definition}
There may be various forms of synchronization restrictions. We will only use the following forms of fundamental restrictions: 
\begin{enumerate}
  \item  $\mathit{simul}(P_1, P_2,i_1,i_2)$: for  processes  $P_1, P_2$ the respective actions at  indices $i_1, i_2$  must happen simultaneously.
  \item  $\mathit{starts}(P_1, P_2,i)$:   process  $P_1 $ starts at index  $i$ of process $P_2$ 
  \item   $\mathit{ends}(P_1, P_2,i)$: process  $P_1$ ends at index  $i$ of process $P_2$.
  \item$\mathit{before}(P_1, P_2, i_1, i_2)$: for  processes  $P_1, P_2$ the action at index $i_1$ of $P_1$ happens simultaneously or before 
     the action at $i_2$  of $P_2$.
\end{enumerate} 
For a set $R$ of restrictions only schedules are permitted that obey all restrictions. This set $R$ is also called a set of {\em basic} restrictions.
 
We also permit Boolean formulas over such basic restrictions.   In this case the permitted schedules must obey the complete formula.
\end{definition}

 Note that in CHF these restrictions correspond to synchronization conditions of: start of a future, waiting for an MVar to be in the right state. 
The simultaneous condition is not necessary for single reduction steps in CHF, but can be used for  blocks of monadic commands.  

We show that there is an algorithm for computing the optimal space and an optimal schedule that has an exponential complexity, where
the exponent is $b\cdot N$ where  $b$ is the size of the Boolean formula and $N$ is the number of processes.

\begin{theorem}\label{theorem:upper-complexity-with-restrictions}
 Let there be  $N$ processes and a set $B$ of Boolean restrictions where $b$ is the size of $B$ and the  size of the input is $n$. 
  Then there is an algorithm to compute the optimal space and an optimal schedule of worst case asymptotic complexity of
 $O(poly(n)\cdot n^{O(b\cdot N)})$, where $\mathit{poly}$ is a polynomial.
\end{theorem}
\begin{proof}
The algorithm is simply a brute force method of trying all possibilities:  For every condition try all tuples of indices. 
The number of different tuples is at most $n^N$ and for trying this for every basic restriction we get an upper bound of $n^{N\cdot b}$. 
Now we have to check whether the time constraints are valid, i.e. there are no cycles, which can be done in polynomial time. 
Now we can split the problem into at most $b+1$ intervals with interception of an index of a condition and apply for every interval the 
algorithm {\SPOPTn}  (see \ref{algorithm:SPOPTn}), which requires time sub-quadratic in $n$ by Theorem \ref{theorem:n-processes-polynomial}.
 Thus we get an asymptotic time complexity as claimed. 
\end{proof}

\begin{corollary}\label{corr:fixed-polynomial-with-restrictions}
 Let there be  $N$ processes and a set $B$ of Boolean restrictions where $b$ is size of $B$ and the  size of the input is $n$. 
 Assume that the number $N$ of processes and the size of $B$  is fixed. 
  Then there is a polynomial algorithm  to compute the optimal space and an optimal schedule. 
\end{corollary}

In general, the optimization problem with synchronization restrictions is NP-complete: 

\begin{theorem}\label{thm:NP-hard}
In the general case of synchronization restrictions, the problem of finding the minimal space is NP-hard and hence NP-complete.
\end{theorem}
\begin{proof}
We use the (perfect) partition problem, which is known to be NP-hard. An instance is a multi-set $A$ of positive integers 
and the question is whether
there is a partition of $A$ into two sub-multi-sets $A_1, A_2$, such that $\sum A_1 = \sum A_2$.

This can be encoded as the question for the minimal space for a scheduling: 
Let $P_i = [0,a_i,0,0]$ for $A = \{a_1,\ldots,a_n\}$ and $P_0 = [0,0,0,0]$, where the indices are $1,2,3,4$.
 The condition is a conjunction of the following disjunctions: 
$(P_0, P_i,2,3) \vee (P_0, P_i,3,2)$.  The optimal space  is reached for a schedule, where indices $1,4$ are zero and 
where at index $2$ and $3$, there is a perfect partition of $A$. 
\end{proof}

\begin{example} We illustrate how an abstract version of the producer-consumer problem can be modeled using \spets{} and synchronization restrictions. 
The idea is that the consumer process $P_1$ produces a list/stream that is consumed by the process $P_2$.  The single elements are also modeled as processes.
Our modelling will be such that the optimal space modelling coincides with the intuition that the space usage of the intermediate list is minimal if
there is an eager consumption of the produced list elements.

We represent the problem as follows. There are two processes $P_1, P_2$, the producer and the consumer, 
which consist of $n$ times the symbol $1$. 
There are also $n$ processes $Q_1, \ldots, Q_n$ that only consist of two elements: a $1$ followed by a $0$, where the processes  represent the unconsumed parts of the exchanged list.
We represent the possible executions by synchronization restrictions: 
\begin{itemize}
  \item $Q_i$ is started by $P_1$ at time point $i$:  $\mathit{starts}(Q_i,P_1,i)$
  \item  $Q_i$ is consumed by $P_2$ at a time point $i$ or later:\\This can be represented by 
      $\mathit{before}(P_2, Q_i, i, 2)$ for all $i$. 
  \item $Q_{i+1}$ ends later  than $Q_i$ for all $i$:    $\mathit{before}(Q_i,Q_{i+1}, 2,2)$   for all $i$.  
\end{itemize}

The start of the space-optimal schedule is as follows and requires 3 units of space: \\\\
\hspace*{1cm}
\begin{tikzpicture}
\node at (0.5,0) {$P_1$};
\node at (7.5,0) {$\dots$};
\node at (0.5,-1) {$P_2$};
\node at (7.5,-1) {$\dots$};
\node at (0.5,-2) {$Q_1$};
\node at (0.5,-3) {$Q_2$};
\node at (0.5,-4) {$Q_3$};
\draw (1,0) -- (7,0);
\draw (1,-1) -- (7,-1);
\draw (1,-2) -- (3,-2);
\draw (2,-3) -- (4,-3);
\draw (3,-4) -- (5,-4);
\foreach \Point in {(1,0), (2,0), (3,0), (4,0), (5,0), (6,0), (7,0), (1,-1), (2,-1), (3,-1), (4,-1), (5,-1), (6,-1), (7,-1)}{\node at \Point {$\mid$};}
\foreach \Point in {(1.5,-0.75), (2.5,-0.75), (3.5,-0.75), (4.5,-0.75), (5.5,-0.75), (6.5,-0.75)}{\node at \Point {$1$};}
\foreach \Point in {(1.5,0.25), (2.5,0.25), (3.5,0.25), (4.5,0.25), (5.5,0.25), (6.5,0.25)}{\node at \Point {$1$};}

\foreach \Point in {(1,-2), (2,-2), (3,-2), (2,-3), (3,-3), (4,-3), (3,-4), (4,-4), (5,-4)}{\node at \Point {$\mid$};}
\foreach \Point in {(1.5,-1.75), (2.5,-2.75), (3.5,-3.75)}{\node at \Point {$1$};}
\foreach \Point in {(2.5,-1.75), (3.5,-2.75), (4.5,-3.75)}{\node at \Point {$0$};}
\end{tikzpicture}
\end{example}

\section{Applications}\label{sec:applications}

\subsection{A Variant of Job Shop Scheduling}\label{sec:jobshop}

A variant of job shop scheduling is the following: 
Let there be $n$ jobs (processes) that have to be performed on a number of identical machines. 
If the focus is on the question how many machines are sufficient for processing, then 
 we can ignore the time and thus only specify the number of machines that are necessary for every single sub-job of any job (process).
 The necessary information is then the list of numbers (of machines) for every job. Note that also the number $0$ is permitted.
The trivial solution would be that all jobs run sequentially, in case the machine lists of every job are of the form 
$[0,k_2, \ldots,k_n,0]$.

If there are in addition (special) time constraints, for example every job starts immediately with a nonzero number of machines, 
and also all jobs end with a nonzero number of machines and they terminate all at the same time,
then our algorithm {\SPOPTn} can be applied in a nontrivial way and will compute the minimal total number of necessary  machines.

In the case of further time constraints, Corollary \ref{corr:fixed-polynomial-with-restrictions} shows that in certain cases there are efficient algorithms
and Theorem \ref{thm:NP-hard} shows that the problem, if there are general time constraints, is NP-complete.

  Our approach and algorithm  is related to resource constrained project scheduling \cite{resource-constrained-book:2008} insofar as we are looking and optimizing 
  the space resource of 
  several given processes (projects). The difference is that in job shop and project scheduling the primary objective is to minimize the overall required time, whereas our
  algorithm computes a minimal bound of a resource (here space) not taking the time into account.

 \subsection{An Implementation for Checking Space Improvements }\label{subsec:chfi}
 
The interpreter CHFi calculates all possible interleavings for CHF-Programs (the program can be downloaded here: 
\text{\url{www.ki.cs.uni-frankfurt.de/research/chfi}}). The interpreter also provides a contrary mode that parallelizes as much as possible. 
We implemented Algorithm~{\SPOPTn}, see  Definition \ref{algorithm:SPOPTn}. It can be used with the eager parallelization mode to calculate the required space for 
independent processes. 
The interpreter can be used to affirm the space improvement property of program transformations for examples and also to falsify conjectures of space improvements
 by comparing the required space returned by the interpreter for the same program before and after the transformation was applied. 
The development of an efficient method to compute  the optimal space consumption and runtime of processes with synchronizations is left for future work.

\section{Conclusion and Future Research}\label{sec:conclusion}
We developed an offline-algorithm {\SPOPTn} that optimizes a given set of parallel and independent processes w.r.t. space and 
computes a space-optimal schedule with runtime $O((N+n)\log N)$ where $n$ is the size of the input and $N$ the number of processes. 
The algorithm is applicable to independent processes in concurrent (lazy-evaluating) languages. 
An application is to find the minimum resources that permit a global schedule in the resource-restricted scheduling projects problem.

\subsection*{Acknowledgements}  We thank David Sabel for discussions and valuable remarks on the subject of this paper.

\bibliographystyle{eptcs}

\end{document}